%% file: main.tex
\documentclass[letterpaper, 10 pt, conference]{ieeeconf}
\IEEEoverridecommandlockouts                              
\title{Risk-Bounded Control with Kalman Filtering \\ and Stochastic Barrier Functions}

\usepackage{algorithm} 
\usepackage{algorithmic}  
\usepackage[algo2e]{algorithm2e} 
\usepackage{hyperref}
\hypersetup{
    colorlinks=true,
    linkcolor=blue,
    filecolor=magenta,      
    urlcolor=purple,
}
\usepackage{lipsum}
\usepackage{graphicx}
\usepackage{subfig}
\usepackage{amsmath}
\usepackage{amsfonts}
\usepackage[dvipsnames]{xcolor}
\newtheorem{Lemma}{Lemma}[section]
\newtheorem{Assumption}{Assumption}[section]

\newtheorem{Problem}{Problem}
\newtheorem{theorem}{Theorem}
\newtheorem{remark}{Remark}
\newtheorem{proposition}{Proposition}
\newtheorem{definition}{Definition}
\newcommand{\xht}{\hat{x}(t)}

\newcommand{\ai}{a}
\newcommand{\bb}{b}
\newcommand{\uc}{\mathsf{u}}
\newcommand{\g}{\mathsf{g}}

\newcommand{\bunderline}[1]{\underline{#1\mkern-3mu}\mkern3mu }
\newcommand{\Zeta}{\mathcal{Z}}
\newcommand{\A}{\mathcal{A}}
\newcommand{\Y}{\mathcal{Y}}
\newcommand{\stcomp}[1]{\widetilde{#1}}

\newboolean{commentsOn}
\setboolean{commentsOn}{False}

\ifthenelse{\boolean{commentsOn}}{\newcommand{\bardh}[1]{{\color{Purple}{#1}}}}
{\newcommand{\bardh}[1]{}}

\author{Shakiba Yaghoubi, Georgios Fainekos, Tomoya Yamaguchi, Danil Prokhorov, Bardh Hoxha
\thanks{S. Yaghoubi, and G. Fainekos 
are with SCAI, Arizona State University, Tempe, AZ, USA.
Email: {\tt \footnotesize <first\_name.last\_name>@asu.edu}}%
\thanks{T. Yamaguchi, D. Prokhorov, and B. Hoxha 
are with the Toyota Research Institute of North America, Ann Arbor, MI, USA.
Email: {\tt \footnotesize <first\_name.last\_name>@toyota.com}}%
\thanks{This research was partially funded by NSF awards OIA 1936997 and CNS 1932068, and DARPA AMP N6600120C4020.}}

\date{February 2021}

\begin{document}

\maketitle
\begin{abstract}
In this paper, we study Stochastic Control Barrier Functions (SCBFs) to enable the design of probabilistic safe real-time controllers in presence of uncertainties and based on noisy measurements. Our goal is to design controllers that bound the probability of a system failure in finite-time to a given desired value. To that end, we first estimate the system states from the noisy measurements using an Extended Kalman filter, and compute confidence intervals on the filtering errors. Then, we account for filtering errors and derive sufficient conditions on the control input based on the estimated states to bound the probability that the real states of the system enter an unsafe region within a finite time interval. We show that these sufficient conditions are linear constraints on the control input, and, hence, they can be used in tractable optimization problems to achieve safety, in addition to other properties like reachability, and stability. Our approach is evaluated using a simulation of a lane-changing scenario on a highway with dense traffic. 
\end{abstract}

\section{Introduction}
As autonomous mobile systems (AMS) are gaining traction in many fields, safety remains a primary concern.
In areas such as personal mobility, an AMS must complete start-to-goal motion tasks while ensuring that no collision with other dynamic or static agents is made.
This is a challenging problem that becomes even more difficult when considering noisy or incomplete sensor information.
Take as an example the problem of controlling an autonomous vehicle in a dense highway crossing from one lane to another.
The sensor information may be noisy and incomplete and therefore various state estimation techniques such as Extended Kalman Filters (EKF) are needed to be utilized.
Furthermore, since this is a safety-critical system, the control to the system should be computed in real time to satisfy desired risk bounds.

In this work, we present a control synthesis method to solve the motion planning problem in uncertain environments and in the presence of partial and noisy measurements while bounding the probability of an upcoming collision to a specified value.
We model the system using Stochastic Differential Equations (SDE) with partial noisy measurements~\cite{oksendal2003stochastic}. 
The true states of the system are estimated using an EKF. We show that if the system has some desired properties \cite{reif2000stochastic}, the estimation will be accurate enough, and the estimation error will be bounded. 
We use the estimation error bounds to compute a safety margin around the unsafe set of system states such that outside this safety margin, expected errors in the estimation would not result in an erroneous classification of an unsafe behavior as a safe one.
Then, we use a Barrier Function (BF) candidate whose level set of value one contains the unsafe set of system states expanded by the safety margin. 
Conditions on this BF candidate that bound its expected value over a finite-time horizon are derived based on the model of the SDE and the EKF. These conditions involve the control inputs, and parameters that control the evolution of the BF expected value.
As the level set of value one of the BF includes the unsafe states and the safety margin, these conditions on the BF provide an upper bound to the value of the risk \cite{kushner1967stochastic,santoyo2019barrier,yaghoubi2020risk}. 
Hence, by constraining the aforementioned parameters and the control inputs, we can bound the risk to a desired threshold.
We combine these constraints with other performance related objectives in a Quadratic Program (QP) that can be solved in real-time to compute control inputs that meet the performance objectives while conforming to the desired risk bounds. As we will explain in the rest of the paper, similar to other Barrier Function methods, the aforementioned QP may become infeasible at some states when constraints on an admissible control input exist, or when multiple safety constraints related to multiple unsafe regions exist in the program.
In spite of this, in many applications such as autonomous driving, generating the best possible solution is necessary even if it does not agree with the constraints. Therefore, we provide a substitute to the previous formulation of the QP that returns the same solution when risk bounds can be met,
otherwise, it provides the least risky action with respect to the given circumstances. 

The motion planning problem with safety guarantees has received significant attention in the past.
Researchers have used methods inspired by reachability analysis in order to generate provably safe trajectories \cite{althoff2008stochastic, malone2017hybrid, kochdumper2020utilizing,leung2020infusing}.
Other data driven methods have also been proposed in ~\cite{chen2018data, sadraddini2018formal}.
However, real-time computation of safe motion plans remains a challenge.
Barrier-based methods have shown promise in this direction since they can ensure forward invariance of a safe set without computation of reachable sets. 
They have been used in a variety of automotive applications such as adaptive cruise control~\cite{ames2014control}, traffic control \cite{xiao2019decentralized} and obstacle avoidance~\cite{chen2017obstacle}. 
They have also seen application in robotics~\cite{glotfelter2017nonsmooth,agrawal2017discrete,emam2019robust}.
When the process involves uncertainty with hard bounds on its magnitude, BFs can be used to verify input-to state safety \cite{kolathaya2018input} as well as safety in the worst case \cite{yaghoubi2020training}. When uncertainty has stochastic characteristics, Stochastic Barrier Functions (SBF) should be utilized to verify safety and other system properties ~\cite{Jadbabaie2007,jagtap2018temporal}.
 Conditions based on SBFs have been used in optimization problems such as QPs to compute control inputs in real-time for stochastic systems. In \cite{clark2019control}, these conditions are designed to maximize the probability of invariance of a set $C$ when measurements are partial and noisy. The derived conditions in \cite{clark2019control} may not be feasible in many applications, and in others they may result in very conservative control actions. The reason is that the conditions are designed to zero out the probability of eventually entering an unsafe set as $t\rightarrow\infty$ which is very conservative for many applications. 
As a result, in this work, we derive certificates on the control inputs based on SBFs to bound the probability of a finite-time failure to a desired value in presence of partial and noisy measurements.

The paper builds on some of the ideas in~\cite{yaghoubi2020risk, clark2019control} to solve the risk-bounded control design problem in presence of process and measurement noise.
The main contributions of the paper are as follows:
1) In Section \ref{sec:sbf_kf}, 
we compute confidence intervals associated with the estimation error of the EKF. Then by considering these bounds on the estimation error, we derive sufficient conditions on ``the control inputs and the parameters that manage the growth rate of the BF'' to bound the risk to a desired value.
2) In Section \ref{secqp}, we utilize these conditions in a QP which can be solved in real-time to bound the risk while considering other performance objectives. 3) We later provide an alternative formulation to the QP to avoid infeasibility when a violation of the risk bounds is inevitable. The solutions of the new QP formulation are identical to that of the primary program if the primary program is feasible.
4) We demonstrate our approach on a highway scenario in which an autonomous vehicle has to traverse through dense traffic with noisy and incomplete data.

\input{body}

\section{Conclusion}
In this paper, we studied Stochastic Control Barrier Functions to design control strategies for stochastic systems that include process noise using partial noisy measurements.
We derive sufficient conditions on the SCBF 
to bound the probability of an imminent undesired behavior (like a collision) in the system to a desired value. 
These conditions preserve probabilistic guarantees under state estimation with extended Kalman Filters assuming that the estimation error is bounded. 
These conditions can be combined with other performance objectives in a quadratic program that can be solved in real-time for designing control strategies. 
The approach is simulated on a lane-changing scenario in a highway with dense traffic. 
In the future, we will apply our method on robotic platforms such as the Human Support Robot (HSR)~\cite{yamamoto2019development} to enable indoor motion planning in the presence of humans. 


\bibliographystyle{plain} 

\bibliography{IEEEabrv,Ref}

\end{document}

%% file: body.tex
\section{Problem Statement}

In this section, we formalize the bounded-risk control problem for stochastic systems with noisy and incomplete state information.

Consider a probability space ($\Omega, \mathcal{F},P)$, and the uncorrelated standard Wiener processes $w(t)$, $v(t)$ defined on this space. A stochastic system is defined using the following Stochastic Differential Equation (SDE) 
\begin{align}
    \label{eq1}
 d{x(t)} &= (f(x(t))+g(x(t)) u(t))dt + G(t)dw(t),\\\label{eq2}
 dy(t) &= Cx(t) dt + D(t) dv(t)
\end{align}
%
%
where $x(t) \in 
\mathbb{R}^{n}$ is a stochastic process, $u(t)\in U \subseteq \mathbb{R}^{l}$ is the control input, $y(t) \in \mathbb{R}^{m}$ is a measurable output, and $f:\mathbb{R}^{n} \rightarrow \mathbb{R}^{n}$, and $g:\mathbb{R}^{n} \rightarrow \mathbb{R}^{n\times l}$ are locally Lipschitz continuous functions. We also assume that functions $f,g$ and the control input $u$ satisfy appropriate conditions such that the differential equations (\ref{eq1}), and (\ref{eq2}) have unique solutions in the proper stochastic sense \cite{oksendal2003stochastic}.

Assume that an unsafe set of states for the system of Eq.~(\ref{eq1}) can be defined using a locally Lipschitz function $h:  \mathbb{R}^n\rightarrow \mathbb{R}_+$ as follows:
\begin{equation}\label{unsaf}
X_{\uc} =: \{ x \;|\;h(x)\leq 0 \}.  \end{equation}
Given the state of the system at time $t$, $x(t)$, and a planning time horizon $T$, we define $p_\uc$ as the probability that the stochastic process $x(\tau), \;t\leq \tau\leq t+T$ enters the unsafe set during this planning horizon, namely,
\begin{align}\label{pu}
    p_\uc &= Pr\{ x (\tau)\in X_\uc \;\mbox{ for some }\small{t\leq \tau\leq t+T}\} \nonumber\\&=  Pr\{ \underset{t\leq \tau \leq t+T}{\inf}  h(x (\tau))\leq 0\} .
\end{align}

Similar to \cite{yaghoubi2020risk}, here, we use the term ``risk" informally to refer to this event's probability ($p _\uc$). 

Hence the problem we need to address is formalized as follows:
\begin{Problem} \label{prb1}
 Find a control policy
 for the system (\ref{eq1}) that at any time $t\geq 0$ maps the sequence of outputs $\{y(t'): t'<t\}$ to a control input $u(t)$ that bounds the risk $p_\uc$ by a desired upper threshold $\bar{p}$, i.e., $p_\uc \leq \bar{p}$.
 \end{Problem}
 
We note that in this problem formulation, the state of the system may be partially observable and affected by noise.

\begin{remark}
In robotic applications where multiple agents need to be modelled to design control policies for a robot, it is useful to separate the model of the robot from the model of the agents as we did in \cite{yaghoubi2020risk} rather than integrating them all in an equation like (\ref{eq1}). This will allow us to consider a varying number of agents around the robot and assign a different desired upper threshhold to the risk associated with each of them.
\end{remark}


\section{Stochastic Barrier Functions Under Kalman Filtering}
\label{sec:sbf_kf}

In this section, we first review some background information about stochastic systems and processes, and Extended Kalman filters and their associated error bounds, and then derive conditions on a BF candidate to bound the risk despite errors in estimation.

As Lie derivatives study the evolution of a function of a deterministic variable, the infinitesimal/differential generators study the evolution of the expectation of a function of a stochastic variable $x(t)$
\cite{Jadbabaie2007}:
\begin{definition}[Infinitesimal/differential generator]\label{generator}
The infinitesimal/differential generator $A$ of a stochastic process $x(t)$ on $\mathbb{R}^{n }$ is defined by
\begin{equation*}\vspace{3pt}
    AB(x_0) = \lim_{t\rightarrow 0}\scalebox{1.1}{$ \frac{E[B(x(t))\;|\;x (0)= x_0]-B(x_0)}{t}$},
\end{equation*}
for all the functions $B:\mathbb{R}^{n }\rightarrow \mathbb{R}$ for which the above limit exists for all $x_0$ \cite{oksendal2003stochastic}.
\end{definition}

For a stochastic process $x(t)$ satisfying Eq. (\ref{eq1}), the differential generator $A$ of a twice differentiable function $B:\mathbb{R}^{n }\rightarrow \mathbb{R}$ is given by \cite{oksendal2003stochastic}
\begin{equation*}
    AB(x) = \frac{\partial B}{\partial x} F(x,u) + \frac{1}{2} tr \Big(G(t)^\top \frac{\partial ^2B}{\partial {x}^2} G(t)\Big).
\end{equation*}
where $F(x,u) = f(x)+g(x) u$, and $tr(.)$ computes the trace of a square matrix.

\subsection{Extended Kalman Filter Estimation}

Due to its appealing properties like ease of implementation, Extended Kalman Filter (EKF) is the most widely used estimator in practical applications for estimating states of the system of Eq. (\ref{eq1}) based on measurements in Eq. (\ref{eq2}). 

\begin{definition}[Extended Kalman Filter]\label{EKF}
An EKF is given by the following equations \cite{reif2000stochastic,brown2012introduction}:
\begin{itemize}
\item Initialization
         \begin{equation}\label{est}
        \hat{x}(0) = E(x(0)),\; P(0) = Var(x(0))
    \end{equation}
    \item Differential equation for the state estimate:
    \begin{equation}\label{est2}
        d\hat{x}(t) = F(\hat{x}(t),u(t))dt +K(t)(dy(t)-C\hat{x}(t) dt)
    \end{equation}
    \item Riccati differential equation:
    \begin{align}
        dP =& [A(t) P(t) +P(t) A(t) ^\top +Q(t)  \nonumber\\
        &- P(t) C ^\top R^{-1}(t)  C P(t)] dt
    \end{align}
    \item Kalman gain:
    \begin{equation}
        K(t) = P(t) C^\top R^{-1} (t)
    \end{equation}
\end{itemize}
    where $A(t) = \frac{\partial F}{\partial x}(\hat{x}(t),u(t))$, $Q(t)$ is a time-varying symmetric positive-definite matrix like $Q(t) = G(t) G(t)^\top$ and $R(t)$ is a time varying positive definite matrix like $R(t) = D(t)D(t)^\top$.
\end{definition}
\subsection{Error Bounds for the Extended Kalman Filter}
In order to use the state estimations $\xht$ for designing control policies that bound ``the probability of the true states $x(t)$ entering an unsafe region'', the estimation error needs to be bounded. 
As we will discuss in the following, to prove boundedness of the estimation error $x(t) - \xht$ using EKF some conditions need to be satisfied:

\begin{Assumption}\label{asm1}
We assume that the system of equations (\ref{eq1}), and (\ref{eq2}) satisfy the following conditions:
\begin{enumerate}
     \item The pair $[\frac{\partial F}{\partial x}(x,u), C], x \in \mathbb{R}^{n}, u\in U$ is uniformly detectable, i.e, there exist a bounded matrix valued function $\Lambda(x)$, and $\gamma>0$ such that $\forall \omega, x \in  \mathbb{R}^{n}, u \in \mathbb{R}^{l}, t\geq 0$
    \begin{equation}
        \omega^\top \Big(\frac{\partial F}{\partial x}(x,u)+\Lambda(x)C \Big)\omega \leq -\gamma \|\omega\|^2
    \end{equation}
    \item There exist $q,r >0$ such that $\forall t \geq 0, qI\leq Q(t)$, and $r I \leq R(t)$.
    \item Let $\phi$ be defined by $F(x(t),u(t)) - F(\xht, u(t)) = A(t)(x(t) - \xht)+\phi(x(t),\xht ,u(t))$. Then there exist $\epsilon_\phi,k_\phi$ such that:
    \begin{align}
        \|\phi(x(t),\xht ,u(t))\|\leq k_\phi \|x(t)-\xht\|^2
    \end{align}
    for all $x, \hat x , u$, with $\|x(t)-\xht\|\leq \epsilon _\phi$. 
    \end{enumerate}
    \end{Assumption}
\begin{proposition}[\cite{reif2000stochastic}]
 If the pair $[\frac{\partial F}{\partial x}(x,u), C]$ is uniformly detectable, then there exist real numbers $\bar \rho , \bunderline \rho>0$ such that the solution to the Riccati differential equation $P(t)$ satisfies: $\bunderline \rho I\leq P(t)\leq \bar \rho I$.
 \end{proposition}

In the following Lemma, assuming that the above conditions are met, we show that for any confidence level $1-p_e$, there exists a confidence interval to which the supremum of the estimation error belongs with the confidence level $1-p_e$,
if the initial estimation error and the process and measurement noises are small enough.
\begin{Lemma}
Consider the system of equations (\ref{eq1}), and (\ref{eq2}), and the EKF defined in Def. \ref{EKF}. If conditions of Asm. \ref{asm1} are met, there exist $\epsilon_0>0,\delta>0$, such that if $\|x(0)-\hat x(0)\|\leq \epsilon_0$, and $G(t)G^\top(t)\leq \delta I$, and $D(t)D^\top(t)\leq \delta I$, then for any $p_e$ there exist $\epsilon>0$ such that:
\begin{equation}\label{epgama}
    Pr\{\underset{t\geq 0}{\sup}\;\|x(t) - \xht\| \leq \epsilon \} \geq 1-p_e
\end{equation}
    \end{Lemma}
\begin{proof}
 Define $\zeta(t) = x(t)-\hat x(t)$. Based on \cite[Appendix]{reif2000stochastic} when $\epsilon_0 = \min (\epsilon_\phi, \frac{q\bunderline\rho}{4k_{\phi}\bar \rho^2})$, $\delta = \frac{q\bunderline\rho \bar \epsilon^2 }{4k_{\phi}\bar \rho^2}$, where $\bar \epsilon$ is a lower bound to the estimation error ($\bar \epsilon \leq \|\zeta(t)\|$), the differential generator of the random process $V(\zeta(t), t) =\zeta(t)^\top P^{-1}(t) \zeta(t)$ satisfies $A V(\zeta(t), t) \leq 0$ hence $V(\zeta, t)$ is a supermartingale. As a result, using the Doob's martigale inequality we obtain:
\begin{align}
    &Pr\{\underset{t\geq 0}{\sup}\;\|x(t) - \xht\|\geq \epsilon\} = Pr\{\underset{t\geq 0}{\sup}\;\zeta(t) ^\top \zeta(t) \geq \epsilon^2\}\leq\nonumber\\ & Pr\{\underset{t\geq 0}{\sup}\;V(\zeta(t), t)\geq \frac{\epsilon^2}{\bar \rho}\}
   \leq \frac{\bar \rho}{\epsilon^2} V(\zeta(0), 0) \leq \frac{\bar \rho \epsilon_0^2}{\bunderline \rho \epsilon^2}
\end{align}
Hence for any $p_e$, one can choose $\epsilon = (\frac{\bar \rho \epsilon_0^2}{\bunderline \rho p_e})^{1/2}$, to complete the proof.
\end{proof}
\vspace{5pt}

\subsection{Bounded Risk Using Stochastic Control Barrier Functions Under Kalman Filtering}

In this section, we present the main result of the paper. In the following, we derive conditions on the estimated state that if satisfied the probability that ``the true states of the system enter the unsafe set $X_\uc$ within a finite time interval $[t,t+T]$'' becomes bounded by some desired value $\bar p$.
To achieve this, similar to the work in \cite{clark2019control}, we first construct a safety margin around the zero level set of the function $h$ such that outside this safety margin, estimation errors of up to size $\epsilon$ will not result in an unsafe behavior.
The size of this safety margin can be computed by mapping the error bounds from the state space to the space of the function $h$ using the following definition:
\begin{equation}\label{12}
    h_\epsilon = \sup \{h(x) : \|x-x'\|\leq\epsilon  \mbox{ for some $x'$ s.t. } h(x') \leq 0\}
\end{equation}

As pictured in Fig. \ref{fig1}, $h_\epsilon$ defines a safety margin around the zero level set of $h$ that based on the following Lemma can compensate for errors of up to size $\epsilon$. 
\begin{Lemma}\label{lem}
Consider any time $t$ such that $\|x(t) - \hat x(t)\|\leq \epsilon$. If $h(\hat x (t))>h_\epsilon$, then $h(x(t))>0$, i.e, $x(t)\notin X_\uc$.
\end{Lemma}

\begin{proof}
Assume by contradiction that $x(t)\in X_\uc$.
Hence, $h(x(t))\leq 0$. By assumption $\|x(t) - \hat x(t)\|\leq \epsilon$, and we have:
\begin{align*}
  & h(\hat x(t))\leq \sup\{h(x): \|x - x(t)\|\leq \epsilon\}\\ &\leq \sup\{h(x): \|x - x'\|\leq \epsilon \mbox{ for some $x'$ s.t } h(x')\leq 0\}\\ 
  &= h_\epsilon
\end{align*}
which contradicts the assumption that $h(\xht)> h_\epsilon$. Hence $h(x(t))>0$, and, $x(t)\notin X_\uc$.
\end{proof}

\begin{figure}[t]
  \centering
        \includegraphics[width=.7\linewidth]{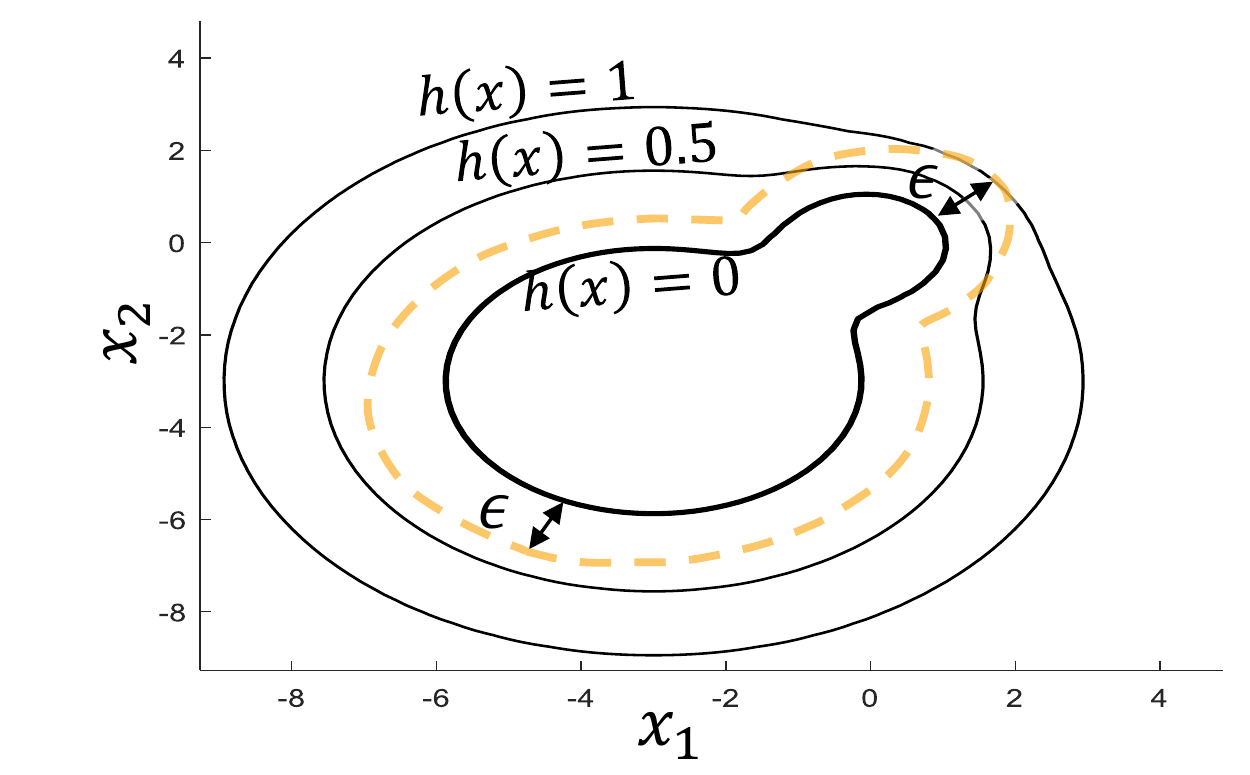}
  \caption{Given a bound $\epsilon$ on the error $x-\hat x$, $h_\epsilon$ can be computed using Eq. (\ref{12}). In this example $h_\epsilon = 1$.}
  \label{fig1}
\end{figure}

Let us show the complement of a set $\A$ with $\stcomp \A$. Define $\Y= \{ \hat x \;|\;h(\hat x)\leq h_\epsilon \}$, and $\Zeta = \{ \zeta \;|\;\|\zeta\|\leq \epsilon \}$. Based on Lemma \ref{lem} we have: $\stcomp\Y \cap \Zeta \subset \stcomp{X_\uc}$, and hence $Pr\{X_\uc\} \leq 1- Pr\{ \stcomp\Y \cap \Zeta \} = Pr\{ \Y \cup \stcomp \Zeta \} = Pr\{ \Y \cap  \Zeta \} + Pr\{\stcomp \Zeta\} = Pr\{ \Y | \Zeta \}  Pr\{\Zeta \} + Pr\{\stcomp \Zeta\}  $. As a result defining $\bar h(x) = h(x)-h_\epsilon$, and $\hat p_\uc = Pr\{ \underset{t\leq \tau \leq t+T}{\inf}  \bar h(\hat x (\tau))\leq 0 \;|\; \underset{t\geq 0}{\sup} \|x(t) - \xht\|\leq \epsilon \}$, we have:
\begin{align}\label{pb}
   &p_\uc
   =Pr\{ \underset{t\leq \tau \leq t+T}{\inf}  h(x (\tau))\leq 0\}\leq  (1-p_e)\hat p_\uc +p_e
\end{align}

Therefore, by designing a control law that guarantees $\hat p_\uc \leq \frac{\bar p -p_e}{1-p_e}$, we also solve Problem \ref{prb1}. 
Similar to our work in \cite{yaghoubi2020risk}, we use this computed upper bound to derive conditions on the evolution of a barrier function candidate $B$ to certify $p_\uc$ is bounded to $\bar p$. The conditions create linear constraints on the control actions based on which the control inputs are computed in real-time.
\begin{definition}\label{def3}
A twice differentiable function $B:\mathbb{R}^n\rightarrow \mathbb{R}_
+$ is a Barrier Function (BF) candidate for the SDE (\ref{eq1}), and (\ref{eq2}) w.r.t the set $X_\uc $, if 
\begin{align}
&B(x)\geq 0 \quad \forall x \in  \mathbb{R}^n,\mbox{ and}\\
    &B(x)\geq 1 \quad \forall x \in  \{ x \;|\;\bar h(x)\leq 0 \}.
\end{align}
\end{definition}

The difference here with \cite{yaghoubi2020risk} is that the BF candidate as defined above should create a barrier around the $h_\epsilon$-level set of $h(x)$ instead of its zero-level set. For a differentiable function $h$, we can choose $B(x) = e^{-\alpha \bar h(x)}$ where $\alpha$ is a positive real number.

As discussed before, in what follows we assume that the functions $f,g$, and the control input $u$ satisfy appropriate conditions to guarantee existence of unique solutions to the differential equations (\ref{eq1}), and (\ref{eq2}). Bounded control inputs $u\in U$ satisfy such appropriate conditions \cite{oksendal2003stochastic}.

\begin{theorem}\label{th2}
Consider the SDE in Eq. (\ref{eq1}), and (\ref{eq2}), and suppose there exist a BF candidate $B(x)$, positive variables $\ai,\bb  \geq 0$, and the control input $u(t)\in U$ s.t. the  condition: 
\begin{align}\label{b1}
 &\frac{\partial B}{\partial x} F(\hat{x}(t),u(t))+ \epsilon \| \frac{\partial B}{\partial x} K(t)C\| + \nonumber\\
 &\frac{1}{2} tr (D(t)^\top K(t)^\top \frac{\partial ^2B}{\partial {x}^2} K(t) D(t)) 
\leq -\ai  B(\hat x) + \bb,
\end{align}
on the function $B$, and one of the following conditions on the variables $a,b$ 
\begin{align}\label{cond3}
    \ai  &= 0, \bb  \leq (p_{new}-B_0)/T, \\\label{cond3.5}
    \ai  &> 0, \bb  \leq \scalebox{1}{$\min(\ai ,-\frac{1}{T}\ln{\frac{1-p_{new}}{1-B_0}})$},\mbox{ or}\\ \label{cond4}
    \ai  &>0, \scalebox{1}{$\frac{\bb (e^{\bb  T}-1)}{p_{new} e^{\bb  T}-B_0}$}\leq \ai \leq \bb 
\end{align}
are satisfied $\forall \xht \in \mathbb{R}^n$, where $B_0 = B(\xht), \; p_{new}=\frac{\bar p -p_e}{1-p_e}$, and Eq. (\ref{epgama}) holds for $p_e,\epsilon>0$. Then for all $t\geq0$, $p_\uc\leq \bar p$.
\end{theorem}

\begin{proof}
From Eq. (\ref{eq2}), and (\ref{est}), we have:
 \begin{align*}
 &\scalebox{1}{$d\hat{x}(t) =$} \\
 &\scalebox{0.96}{$ F(\hat{x}(t),u(t))dt +K(t)\Big(Cx(t) dt + D(t) dv(t)-C\hat{x}(t) dt\Big)=$}
 \\&\scalebox{.96}{$\Big(F(\hat{x}(t),u(t))+K(t)(C\hat{x}(t)-Cx(t))\Big)dt +K(t)D(t) dv(t)$}
\end{align*}
and hence from Def. \ref{generator} of the infinitesimal generator:
\begin{align*}
 & \scalebox{1}{$AB(\hat{x}(t)) = \frac{\partial B}{\partial x} \Big( F(\hat{x}(t),u(t))+K(t)(C\hat{x}(t)-Cx(t))\Big)$}
 \\&+\scalebox{1}{$\frac{1}{2} tr (D(t)^\top K(t)^\top \frac{\partial ^2B}{\partial {x}^2} K(t) D(t)) $} 
\end{align*}
Assuming $\|x(t)-\hat x(t)\|<\epsilon$, we have $\frac{\partial B}{\partial x} K(t)C(\hat{x}(t)-x(t)) \leq \epsilon  \|\frac{\partial B}{\partial x} K(t)C\| $, and hence:
\begin{align*}
 AB(\hat{x}(t)) \leq \frac{\partial B}{\partial x} F(\hat{x}(t),u(t))+ \epsilon  \|\frac{\partial B}{\partial x} K(t)C\| +\\
 \frac{1}{2} tr (D(t)^\top K(t)^\top \frac{\partial ^2B}{\partial {x}^2} K(t) D(t)).  
\end{align*}
So if Eq. (\ref{b1}) holds, then $AB(\xht)\leq -aB(\xht)+b$, and, hence, if one of (\ref{cond3}-\ref{cond4}) holds, based on \cite[Thm. 2]{yaghoubi2019worst}
we have $ \hat p_\uc\leq \bar p_{new}$, and from Eq. (\ref{pb}), we get $p_\uc \leq \bar p$.
\end{proof}

A BF candidate $B$ that satisfies the conditions of Thrm. \ref{th2} is a risk-bounding Stochastic Control Barrier Function (SCBF). Any control policy that satisfies conditions of Thrm. \ref{th2} at any time $t\geq0$ for some SCBF bounds $p_\uc$ to $\bar p$.

\section{Risk-Constrained Optimization-Based Control Design}\label{secqp}
Consider a system of the form (\ref{eq1}), and (\ref{eq2}) for which the performance objective is expressed by minimizing the expected value of the objective function $J(u(t),x(t)) = \frac{1}{2}u^\top(t) Q(x(t)) u(t)+H(x(t))^\top u(t)$ where $\forall t\;:\: Q(x(t))\in \mathbb{R}^{l\times l}$ is a positive definite matrix, and $H(x(t))\in  \mathbb{R}^{l}$. In order to design risk constrained controllers, this performance objective can be unified with constraints of Thrm. \ref{th2} in the following program:
\begin{align}\label{prog}
       \min_{u(t)\in U, \ai,\bb}&J(u(t),\hat x(t))\\
\mbox{ s.t.}& \begin{cases}
    \mbox{Ineq. (\ref{b1})}\\
    \mbox{Ineq. (\ref{cond3}), or (\ref{cond3.5}), or (\ref{cond4})} 
    \end{cases}. \nonumber
\end{align}

The constraint of Eq. (\ref{b1}) is linear in the control input $u$, and parameters $a,b$.  
So as to achieve a Quadratic program (QP) that can be solved using efficient solvers, if the second constraint is imposed by Eq. (\ref{cond3.5}), or Eq. (\ref{cond4}), $\ai$ or $b$ need to be fixed to positive values respectively. Then, the program (\ref{prog}) which becomes a QP can be solved by searching over the second parameter and the control input value $u$ to find an optimal policy that bounds the risk to $\bar p$. The program can be extended to account for multiple unsafe sets $X_{\uc,i} =:\{x\;|\; h_i(x)\leq 0\}$ by adding constraints of the form (\ref{b1}), and one of (\ref{cond3}), (\ref{cond3.5}), or (\ref{cond4}) for each unsafe region.  As it is also discussed in \cite{yaghoubi2020risk}, to avoid unnecessary risk when possible, parameters $a$ and $b$ can be added to the objective function with negative and positive multipliers respectively (note that both larger $a$ values, and smaller $b$ values represent more conservative actions).

\subsection{Feasibility of the problem}
Existence of multiple safety constraints w.r.t multiple unsafe sets in addition to the bounds on the control input $u\in U$ may result in an infeasible program (\ref{prog}) at some states $x(t)$ -- or their corresponding estimates $\xht$. Consider for instance an autonomous vehicle that is put into a dangerous situation by another agent. This agent might have unsafely steered into the autonomous vehicle's lane, or hit the brake in a high-speed highway. Given the state of the autonomous vehicle, bounding the risk to its desired value may not be possible. In practice, even if the system cannot limit the risk to its desired value, it should still apply the best control policy that minimizes the risk. To achieve this, we can add a slack variable $s$ to the right hand side of Equations (\ref{cond3}), and (\ref{cond3.5}), or use it to loosen the upper and lower limits of $a$ in Eq. (\ref{cond4}), as follows:
\begin{align}\label{condmod1}
    \ai  &= 0, \bb  \leq (p_{new}-B_0)/T+s, \\\label{condmod2}
    \ai  &> 0, \bb  \leq \scalebox{1}{$\min(\ai ,-\frac{1}{T}\ln{\frac{1-p_{new}}{1-B_0}})$}+s,\mbox{ or}\\ \label{condmod3}
    \ai  &>0, \scalebox{1}{$\frac{\bb (e^{\bb  T}-1)}{p_{new} e^{\bb  T}-B_0}$} -s \leq \ai \leq \bb +s
\end{align}
and modify program (\ref{prog}) as follows:
\begin{align}\label{prog2}
       \min_{u(t)\in U, \ai,\bb, s}&J(u(t),\hat x(t))+cs\\
\mbox{ s.t.}& \begin{cases}
    \mbox{Ineq. (\ref{b1})}\\
    \mbox{Ineq. (\ref{condmod1}), or (\ref{condmod2}), or (\ref{condmod3})} \\
    s\geq 0
    \end{cases}. \nonumber
\end{align}
where $c>\!>0$ is a constant variable. Note that when $c\rightarrow \infty$, and program (\ref{prog}) is feasible, the solution to program (\ref{prog2}) is the same as the solution to program (\ref{prog}) since $s$ can be chosen as zero to minimize the cost function while constraints are still satisfied. In addition, when program (\ref{prog}) is infeasible, program (\ref{prog2}) returns the least unsafe feasible control policy as $c>\!>0$.

\section{Case Study}

Consider an ego vehicle modeled by a unicycle model as follows:
\begin{equation}\label{uni-model}
  d {x}_r(t) =  g_r(x_r(t))u(t)dt = \scalebox{0.9}{${\small \begin{bmatrix}
    \cos(\theta_r) \quad 0\\
    \sin(\theta_r) \quad 0\\
   \; 0 \quad \qquad 1
    \end{bmatrix}}\begin{bmatrix}
   u_1\\
   u_2
    \end{bmatrix}$}dt. 
\end{equation}
where $x_r = [\mathsf{p}_r,\theta_r]^\top =[\mathsf{p}_r^x,\mathsf{p}_r^y,\theta_r]^\top$, $u = [u_1 , u_2]^\top$ consist of the ego vehicle's states and inputs, $\mathsf{p}_r^x,\mathsf{p}_r^y,\theta_r$ describe the x and y position of the vehicle and its heading angle respectively, and $u_1, u_2$ are its linear and angular velocities. The initial condition of the ego vehicle is $x_r(0) =[0,0,0]^\top$ and its inputs are considered to be constrained in the set $U = \{0.2\leq u_1\leq 2, \; -\pi/6\leq u_2\leq \pi/6\}$. Assume that the ego vehicle is in a highway scenario in which other traffic participants are modeled using the following SDE:
\begin{align}\label{traffic}
  d {x}_o(t) &= \scalebox{1}{${\small \begin{bmatrix}
   v_o\\
    0\\
    v_d+ (\mathsf{p}_{o}^x-\mathsf{p}_r^x)e^{(c_1-c_2\|\mathsf{p}_{o}-\mathsf{p}_r\|^2)}-v_o
    \end{bmatrix}}$}dt+ Gdw(t)\\\label{mtraffic}
 dy_o(t) &= Cx_o(t) dt + D dv(t)    . 
\end{align}
where the agents' state $x_{o} = [\mathsf{p}_{o}, v_o]^\top = [\mathsf{p}_{o}^x,\mathsf{p}_{o}^y, v_o]^\top$ describes their position $\mathsf{p}_{o}$, and its velocity $v_o$. The agents try to maintain the desired velocity $v_d$ of the highway but they are also equipped with a mechanism to prevent accidents with the ego vehicle when it is in their vicinity ($c_1, c_2$ are positive constants). We assume that the state of the ego vehicle $x_r$ is known, and the process noise in Eq. (\ref{traffic}) ($G = 0.1\times I_3$) includes possible inaccuracies in estimation of $x_r$ from the other agents' perspectives. However, the ego vehicle can only measure the positions of other agents ($\mathsf{p}_{o}$), i.e,
\scalebox{1}[0.8]{$C = \begin{bmatrix}
   1 \;0\; 0 \\
    0 \;1\; 0
    \end{bmatrix}$}, and there is also an error associated with this measurement (\scalebox{1}[0.8]{$D = \begin{bmatrix}
   0.25 \;0 \\
    0 \;0.2
    \end{bmatrix}$}).
    
Control policies $u_1$ and $u_2$ should be designed so as to bound probabilities of an ego car's imminent accident with other traffic participants ($p_\uc$ as defined in Eq. (\ref{pu})) to 0.1 ($T = 1$, and $h(x_r,x_o) = \|\mathsf{p}_{x}-\mathsf{p}_{o}\|^2-r_\uc^2$, $r_\uc = 0.25$). At the same time, control policies should guide the ego vehicle to the goal set 
\begin{equation}\label{G-set}
    X_\g =: \{x_r\;|\; \|\mathsf{p}_r- x_g\|^2 - r_g^2 \leq 0 \},   
\end{equation}
where $x_g$ is the center and $r_g$ is the radius of the goal set. In our simulation, the goal set is the center of the top-most lane in the highway as depicted in Fig. \ref{fig:inits} with green color.

\begin{remark}
Constraints imposed by a control lyapunov like function can be added to the program (\ref{prog2}) to lead the states $x_r$ to a goal set $X_g$ (for more details see \cite{yaghoubi2020risk}). 
\end{remark}

The control inputs $u_1$ and $u_2$ (the linear and angular velocities) in Eq. (\ref{uni-model}) have different relative degrees w.r.t $h$. Hence to avoid involved control design methods that tackle this issue, we use a similar approach to \cite{yaghoubi2020risk,lindemann2020control} and use a near-identity diffeomorphism to approximate the system of Eq (\ref{uni-model}). Define $ \bar x_r =[\bar{\mathsf{p}}_r,\theta_r]^\top$, where
$\bar {\mathsf{p}}_r := {\mathsf{p}}_r +l
    R(\theta_r)$, $l>0$ is a small constant that allows for approximating $\mathsf{p}_r$ with the needed precision with $\bar {\mathsf{p}}_r$, and $R(\theta_r) = \small{\begin{bmatrix}
     \cos(\theta_r) \quad -\sin(\theta_r)\\
    \sin(\theta_r) \qquad \cos(\theta_r)
    \end{bmatrix}}$. Hence:
\begin{equation}\label{trans}
    \dot {\bar x}_r (t)= {\small\begin{bmatrix}
    \cos(\theta_r) \quad -l\sin(\theta_r)\\
    \sin(\theta_r) \qquad l\cos(\theta_r)\\
   \; 0 \qquad \qquad 1
    \end{bmatrix}}\begin{bmatrix}
   u_1\\
   u_2
    \end{bmatrix}. \vspace{-2pt}
\end{equation}
Note that the maximum distance of $x_r$ from $\bar x_r$ is $l$. Hence, we can redefine $h$ as $h(\bar x_r, x_o)=\|\bar {\mathsf{p}}_r-\mathsf{p}_{o}\|^2-(0.25+l)^2$ to account for the approximation error.

We consider $15$ traffic participants around the ego vehicle in the highway that are initially positioned as pictured in the top subplot of Fig \ref{fig:inits}. Their initial velocities are picked from a normal distribution with mean $v_d$ and variance 0.1. They are modelled based on the SDE (\ref{traffic}), and the ego vehicle has access to noisy measurements $y_o$ based on which it should pick the control inputs $u_1,u_2$. The states of other traffic participants $x_o$ were estimated using an EKF. We selected the parameters $p_e$, and $\epsilon$ based on simulations of the traffic participants in presence of the ego vehicle. In these simulations the estimation error was less than 0.5 for over 99\% of the data. Hence, we took $p_e = 0.01$, and $\epsilon = 0.5$. Noticing the definition of $h(\bar x_r,x_o)$ in our case study, and based on Eq. (\ref{12}), we set $h_\epsilon = \epsilon^2$. So we chose the barrier function candidate as $B(\bar x_r,\hat x_o) = e^{-\gamma \bar h(\bar x_r,\hat x_o)}$ where $\bar h (\bar x_r,\hat x_o) = h (\bar x_r,\hat x_o) - h_\epsilon$. 

\begin{figure}[t]
  \centering
        \includegraphics[width=1\linewidth]{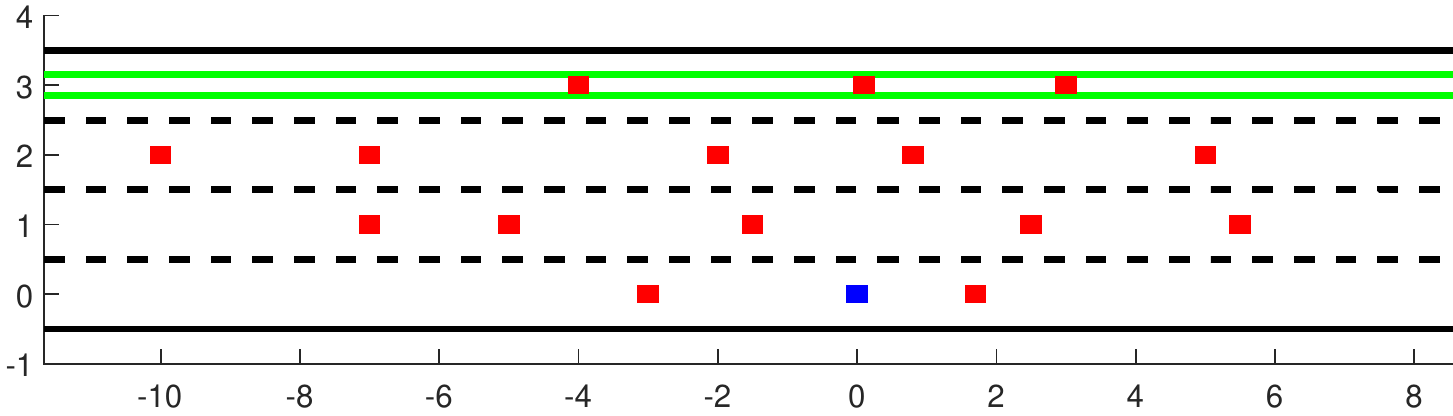}
        \includegraphics[width=1\linewidth]{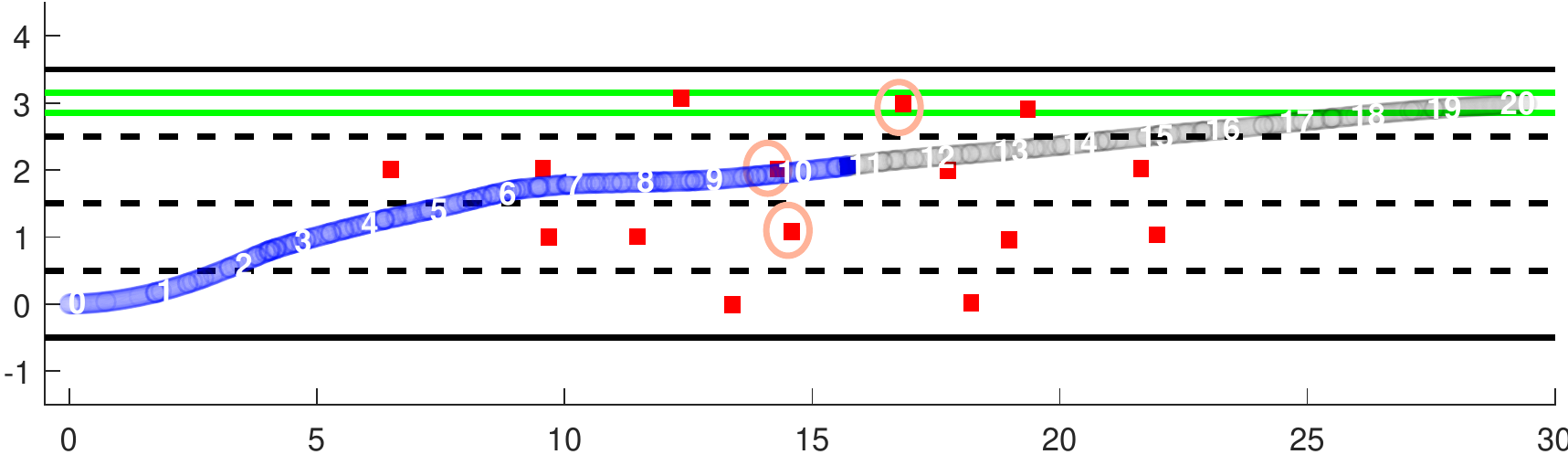}
  \caption{The top figure, shows the initial position of the ego vehicle and traffic participants. The ego vehicle's goal set in the top-most lane is shown in green. The bottom figure is a snapshot of the traffic participants at $t=11$, and the trajectory of the ego vehicle.}
  \label{fig:inits}
\end{figure}
\begin{figure}[t]
  \centering
        \includegraphics[width=1\linewidth]{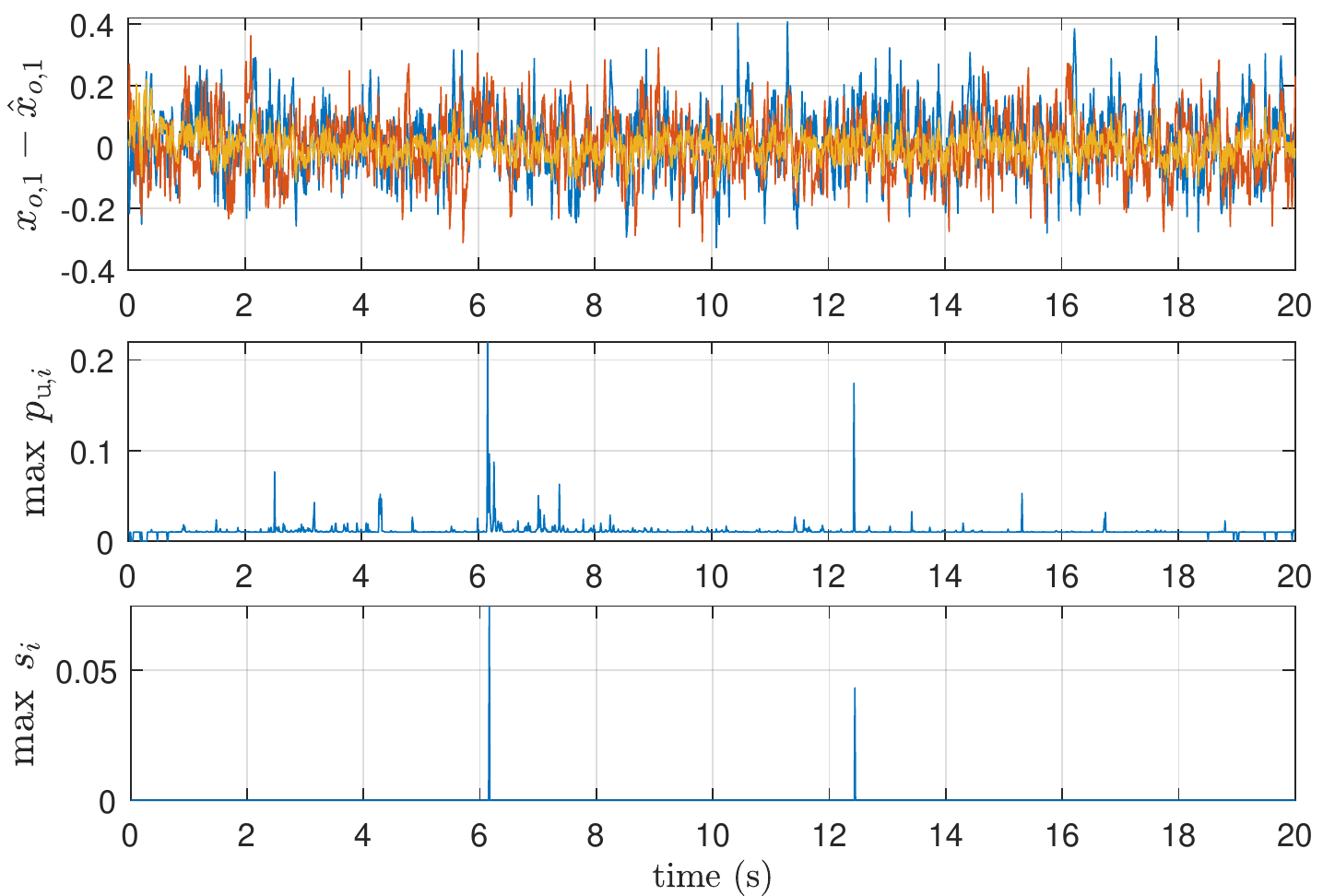}
  \caption{Figures show the error in the state estimation, an upper bound to the risk, and the value of the slack variable in program (\ref{prog2}).}
  \label{fig:plot}
\end{figure}

In order to design the control policies $u$, at each iteration, we estimated the states of the traffic participants $\hat x_o$ based on the measurements $y_o$ using the Kalman filter, and solved program (\ref{prog2}) constrained by a set of inequalities that each corresponds to a close-by traffic participant. Hence, the input $\hat x$ to the program consists of the state of the ego vehicle $\bar x_o$, and the state estimations of the traffic participants close to the ego vehicle (In our case in a distance of less than 2.5). An additional constraint imposed by a control Lyapunov function was also added to program (\ref{prog2}) to lead the ego vehicle toward the goal set. We fixed the parameter $a_i$ to 1, chose $c$ to be a very large value, and added the variables $b_i$ to the objective function to decrease the risk when a risky action can be avoided.

An snapshot of the ego vehicle and the traffic participants' positions in the highway at time $t = 11$ is shown at the bottom subplot of Fig. \ref{fig:inits}. The orange ellipsoidal sets show the contours of Gaussian distributed state estimates $\hat x_o$ with the associated covariance matrix $P$, and the significance level $1-p_e = 0.99$. The subplot also shows the time-stamped trajectory of the ego vehicle. The blue part of the trajectory shows the path of the ego vehicle up to $t= 11$ and the gray part shows the path it will take in the rest of the simulation.
The simulation of the scenario can be found at \href{https://youtu.be/A1EoCq5V3PM}{youtu.be/A1EoCq5V3PM}.
It is worth noting that since in this scenario measurements $y_o$ are received continuously, and the traffic participants have a near linear behavior, the ellipsoidal sets related to the state estimates do not change much in size. 
In the top subplot of Fig. \ref{fig:plot} the estimation error for one of the traffic participants is depicted. As expected the size of these errors match the value of $\epsilon = 0.5$.
The middle subplot shows the overall upper bound to the risk computed by taking the maximum from the upper bounds of the risks $p_\uc$ associated with all the nearby traffic participants. These upper bounds to the risks which we show with $\bar p_{\uc,i}$ for the ith traffic participant is computed as $\bar  p_{\uc,i} = p_e+(1-pe) \bar p_{B,i}$ in which $\bar p_{B,i}$ is an upper bound to $\hat p_\uc$. From \cite{yaghoubi2020risk}, we have $\bar p_{B,i} = 1- (1- B_{0,i}) e^{-\bb _i T}$, where $B_{0,i}$, and $b_i$ are the values of the barrier function and the parameter $b$ corresponding to the current state of the $i$-th traffic participant, respectively. It can be observed that the upper bound to the risk (and hence the risk itself) is bounded to the desired value 0.1 almost everywhere except in 2 time instances. In these 2 time instances it is not possible to find a feasible solution to the problem assuming $s = 0$ (i.e., program (\ref{prog}) is infeasible). As pictured in the bottom subplot, in these time instances the slack variable $s$ has been used to find a feasible solution to program (\ref{prog2}), and in the rest of the time it has a zero value confirming that the solutions to program (\ref{prog}) and (\ref{prog2}) are the same when program (\ref{prog}) is feasible.